\documentclass[]{llncs}
\usepackage{amsmath, amssymb, amscd, amsfonts,xcolor, hyperref, url}
\usepackage{graphicx}
\newtheorem{de}{Definition}

\newtheorem{alg}{Algorithm}

\newtheorem{pro}{Proposition}
\newtheorem{thm}{Theorem}
\newtheorem{rem}{Remark}

\urldef{\mail}\path{r.rastaghi59@gamail.com}
\begin{document}
\title{An Efficient CCA2-Secure Variant of the McEliece Cryptosystem in the Standard Model}
\author{Roohallah Rastaghi}
\institute{Department of Electrical Engineering, Aeronautical University of Since and Technology 
Tehran, Iran\\
\mail}
\maketitle

\begin{abstract}
Recently, a few chosen-ciphertext secure (CCA2-secure) variants of the McEliece public-key encryption (PKE) scheme in the standard model were introduced. All the proposed schemes are based on encryption repetition paradigm and use general transformation from CPA-secure scheme to a CCA2-secure one. Therefore, the resulting encryption scheme needs \textit{separate} encryption and has \textit{large}  key size compared to the original scheme, which complex public key size problem in the code-based PKE schemes. Thus, the proposed schemes are not sufficiently efficient to be used in practice.

In this work, we propose an efficient CCA2-secure variant of the McEliece PKE scheme in the standard model. The main novelty is that, unlike previous approaches, our approach
is a generic conversion and can be applied to \textit{any} one-way trapdoor function (OW-TDF), the lowest-level security notion in the context of public-key cryptography, resolving a big fundamental and central problem that has remained unsolved in the past two decades.
\end{abstract}
\section{Introduction} \label{sec1}
Post-quantum cryptography has obtained great attention in recent years. Code-based cryptography holds a great promise for the post-quantum cryptography, as it enjoys very strong
security proofs based on average-case hardness \cite{22}, relatively fast and efficient encryption/decryption nature, as well as great simplicity.
In the context of code-based cryptography, there are two well-known public-key encryption (PKE) schemes, namely McEliece \cite{13} and Niederreiter \cite{15} PKE schemes. The McEliece encryption scheme was the first PKE scheme based on linear error-correcting codes. It has a very fast and efficient encryption procedure, but it has one big flaw: the size of the public key. Recently, how
to reduce the public key size and how to secure the parameter choice in the code-based cryptography are deeply explored \cite{1,2,3,9,14}.

Semantic security (a.k.a indistinguishability) against adaptive chosen ciphertext attacks (CCA2 security) is one of the strongest known notions of security for the PKE schemes was introduced by Rackoff and Simon \cite{20}. It is possible to produce CCA2-secure variants of the code-based PKE schemes in the random oracle model \cite{4,11,12}, however, CCA2 security in the standard model has not been widely discussed. To the best of our knowledge, only a few papers have touched this research issue.
\subsection{Related work}\label{sec1.1}
There are mainly two class of CCA2-secure code-based PKE schemes in the standard model.
\begin{itemize}
\item {\it CCA-secure schemes based on syndrome decoding problem.} Freeman \textit{et al.} \cite{10} used Rosen-Segev approach \cite{21} to introduce a correlation-secure trapdoor function related to the
hardness of syndrome decoding. Their construction is based on the Niederreiter PKE scheme. Very recently, Preetha Mathew \textit{et al.} \cite{19} proposed a somewhat efficient variant of the
Niederreiter scheme based on lossy trapdoor functions \cite{17}, which avoids encryption repetition paradigm.
\item{\it CCA-secure schemes based on general decoding problem}. The first CCA2-secure variants of the McEliece cryptosystem was introduced by Dowsley \textit{et al.} \cite{5}. They proposed a scheme that resembles the
Rosen-Segev approach trying to apply it to the McEliece PKE scheme. Their construction has some ambiguity. The scheme does not rely on a collection of functions but instead defines a
structure called {\it $k$-repetition} PKE scheme. This is essentially an application of $k$-samples of the PKE to the {\it same} input, in which
the decryption algorithm also includes a verification step on the $k$ outputs. The encryption algorithm produces a signature directly on the McEliece ciphertexts instead of introducing
a random vector $x$ as in the original Rosen-Segev scheme; therefore a CPA-secure variant of the McEliece cryptosystem is necessary to achieve CCA2 security \cite{18}. Very recently, D\"{o}ttling \textit{et al.} \cite{6} showed that Nojima \textit{et al.} \cite{16} randomized version of the McEliece cryptosystem is $k$-repetitions CPA-secure and, as we mentioned earlier, it can obtain CCA2 security by using a strongly unforgeable one-time signature scheme. In a subsequent work, Persichetti \cite{18} proposed a CCA2-secure PKE scheme based on the McEliece assumptions using the original Rosen-Segev approach.
\end{itemize}
\subsection{Motivation}\label{sec1.2}
To date, as we stated above, all the proposed CCA2-secure code-based PKE schemes in the standard model are based on either lossy and correlation-secure trapdoor functions or $k$-repetitions encryption paradigm. Therefore, the resulting encryption schemes are not efficient as they need to run encryption/decryption algorithms several times and use a strongly unforgeable one-time signature scheme to handle CCA2 security related issues. Moreover, in these schemes, excluding the keys of the signature scheme, the public/secret keys are $2k$-times larger than the public/secret keys of the original scheme, which complex the public key length problem in the code-based PKE schemes. Although the Preetha Mathew \textit{et al.}'s scheme \cite{19} avoids $k$-repetitions paradigm, it yet needs to run encryption/decryption algorithms 2-times and the public/secret
keys are larger than the original Niederreiter scheme. Further, it also uses a strongly unforgeable one-time signature scheme to achieve CCA2 security, and so needs separate encryption.
Hence, how to design an efficient CCA2-secure code-based encryption scheme in the standard model is still worth of investigation. This motivates us to investigate new approach for construction efficient such schemes in the standard model without using encryption repetition and generic transformation from CPA-secure schemes to a CCA2-secure one.
\subsection{Our Contributions}
To tackle the above challenging issues, we introduce a randomized encoding algorithm called {\sf PCA} and use it along with the McEliece PKE scheme to construct a CCA2-secure PKE scheme in the standard model. Our contributions in this paper are:
\begin{itemize}
  \item The main novelty is that our construction is a generic conversion and can be applied to any low-level primitive. To further demonstrate the usefulness of our approach, in Section \ref{sec4} we also introduce \textit{direct} ``black-box" construction of a CCA2-secure PKE scheme from any TDF in the standard model, resolving a big fundamental and central problem in the context of public-key cryptography that has remained unsolved in the past two decades.
  \item Our proposed scheme is more efficient, the publick/secret keys are as in the original scheme and the encryption/decryption complexity are comparable to the original scheme.
  \item This novel approach leads to the elimination of the encryption repetition and using strongly unforgeable one-time signature scheme.
  \item This scheme can be used for encryption of long length messages without employing the hybrid encryption method and symmetric encryption.
\end{itemize}
\textbf{Organisation.} In the next section, we briefly explain some mathematical background and definitions. Then, in Section \ref{sec3}, we introduce our proposed scheme. Finally, a generalized construction based on OW-TDFs will be given in Section \ref{sec4}.
\section{Preliminary}\label{sec2}
\subsection{Notation}
We represent a binary string in general by bold face letter such as ${\bf x}=(x_1, \dots x_n)$. \sloppy Regular small font letter $x$ denotes its corresponding \textit{decimal} value, that is $x=\sum_{i=1}^{n}x_i 2^{(n-i)}$ and $\left| \textbf{x} \right|$ denotes its binary length. If $k \in \mathbb{N}$ then $\left\{ {0,\,1} \right\}^k$ denote the set of $k$-bit strings,
$1^k$ denote a string of $k$ ones and $\left\{ {0,\,1} \right\}^*$ denote the set of bit strings of finite length. $y \leftarrow x$ denotes the assignment to \textit{y }of the value
\textit{x}. For a set $S$, $s \leftarrow S$ denote the assignment to $s$ of a uniformly random element of $S$. For a deterministic algorithm ${\cal A}$, we write
$x\leftarrow {\cal A}^{\cal O} (y,\,z)$ to mean that \textit{x} is assigned the output of running ${\cal A}$ on inputs \textit{y} and \textit{z}, with access to oracle ${\cal O}$.
If ${\cal A}$ is a probabilistic algorithm, we may write $x \leftarrow {\cal A}^{\cal O} (y,\,z,\,\,R)$ to mean the output of ${\cal A}$ when run on inputs \textit{y} and \textit{z} with
oracle access to ${\cal O}$ and using the random coins $R$. We denote by $\Pr[E]$ the probability that the event $E$ occurs. If $a$ and $b$ are two strings of bits, we denote by $a\|b$  their concatenation. ${\sf Lsb}_{x_1}(a)$ means the right $x_1$ bits of $a$ and ${\sf Msb}_{x_2}(a)$ means the left $x_2$ bits of \textit{a}.

Since the proposed cryptosystem is code-based, a few notations regarding coding theory are introduced. Let $\mathbb{F}_2$ be the finite field with 2 elements
$\{0, 1\}$, $k\in \mathbb{N}$ be a security parameter. A binary linear-error correcting code ${\cal C}$ of length $n$ and dimension $k$ or an $[n, k]$-code is a $k$-dimensional
subspace of $\mathbb{F}_2^n$. Elements of $\mathbb{F}_2^n$ are called words, and elements of ${\cal C}$ are called codewords. If the minimum hamming distance between any two codewords
is $d$, then the code is a $[n, k, d]$ code. The Hamming weight of a codeword ${\bf x}$, ${\sf wt}({\bf x})$, is the number of non-zero bits in the codeword. For
$t \leq \lfloor \frac{d-1}{2} \rfloor$, the code is said to be $t$-error correcting if it detects and corrects errors of weight at most $t$. Hence, the code can also be represented as a
$[n, k, 2t + 1]$ code. The generator matrix ${\sf G} \in \mathbb{F}_2^{k\times n}$ of a $[n, k]$ linear code ${\cal C}$ is a matrix of rank $k$ whose rows span the code ${\cal C}$.
\subsection{Definitions}
\begin{de}[{\bf Trapdoor functions}] A trapdoor function family is a triple of algorithms ${\sf TDF=(Tdg, F, F^{-1})}$, where {\sf Tdg} is probabilistic and on input $1^k$ generates an evaluation/trapdoor key-pair $(ek, td )\leftarrow {\sf Tdg}(1^k)$. ${\sf F}(ek, \cdot)$ implements a function $f_{ek} (\cdot)$ over $\{0, 1\}^k$ and ${\sf F}^{-1}(td,\cdot)$ implements its inverse $f^{-1}(\cdot)$.\\
\end{de}
\begin{de}[{\bf One-wayness}] Let ${\cal A}$ be an inverter and define its {\sf OW-advantage} against {\sf TDF} as \sloppy
\[
{\rm Adv}_{{\sf TDF},\, {\cal A}}^{ow} (k) = \Pr \left[x = x' :
                                                    \begin{array}{ll}
                                                      (ek , td) \leftarrow  {\sf Tdg}(1^k) ; x \leftarrow  \{0, 1\}^k \\
                                                      y  \leftarrow {\sf F}(ek , x) ; x' \leftarrow  {\cal A}(ek , y)
                                                    \end{array}
                                                  \right].
\]
\end{de}
Trapdoor function {\sf TDF} is one-way if ${\rm Adv}_{{\sf TDF}, {\cal A}}^{ow}(k)$ is negligible for every PPT inverter ${\cal A}$.
\begin{de}[{\bf Circular Shift}] A circular (cyclic) shift is the operation of rearranging the components in a string circularly with a prescribed ‎number of positions. Thus, a $q$-position circular shift (or circular $q$-shift) defines as the operation in which the $i$-th sample, $s_i$, replace with the $(i + q \mod n)$-th sample in a $n$ sample ensemble. We denote this operation by ${\sf CS}_{q,n}(s_i)= s_{(i+q \mod n)}$, $1\leq i\leq n$.
\end{de}
\begin{de}[{\bf General Decoding Problem}]
Given a generator matrix ${\sf G} \in \mathbb{F}_{2}^{k\times n}$ and a word ${\bf m} \in \mathbb{F}_{2}^{n}$, find a codeword $ {\bf c} \in \mathbb{F}_{2}^{k}$ such that $ {\bf e}= {\bf m}-{\bf c}{\sf G}$ has Hamming weight ${\sf wt}({\bf e})\leq t$.
\end{de}
\begin{de} [{\bf General Decoding Assumption}] Let ${\cal C}$ be an $[n, k, d]$-binary linear code defined by a $ k \times n$ generator matrix $\sf G$ with the minimal distance $d$, and
$t \leq \lfloor \frac{d-1}{2} \rfloor$. An adversary ${\cal A}$ that takes an input of a word $ {\bf m} \in \mathbb{F}_{2}^{n}$, returns a codeword ${\bf c} \in \mathbb{F}_{2}^{k}$.
We consider the following random experiment on $GDP$ problem.
\begin{eqnarray*}
&&{\bf Exp}_{\cal A}^{\rm GDP}:\\
&&\quad \quad  {\bf c} \in \mathbb{F}_{2}^{k}\,\, \leftarrow \,\,{\cal A}({\sf G}, {\bf m} \in \mathbb{F}_{2}^{n})\\
&&\quad \quad  {\rm if}\,\,\, {\bf x}= {\bf m}-{\bf c}{\sf G} \,\,{\rm and}\,\, {\sf wt}({\bf x})\leq t\\
&& \quad \quad {\rm then} \ b \leftarrow 1, \ {\rm else} \ b \leftarrow 0\\
&& {\rm return}\ b.
\end{eqnarray*}
We define the corresponding success probability of ${\cal A}$ in solving the GDP problem via
\[
{\bf Succ}_{\cal A}^{\rm GDP} = \Pr [{\bf Exp}_{\cal A}^{\rm GDP}= 1].
\]
Let $\tau \in \mathbb{N}$ and $\varepsilon \in [0, 1]$. We call GDP to be $(\tau, \varepsilon)$-secure if
no polynomial algorithm ${\cal A}$ running in time $\tau$ has success ${\bf Succ}_{\cal A}^{\rm GDP} \geq \varepsilon$.
\end{de}
\begin{de} [{\bf Public-key encryption}] A public-key encryption (PKE) scheme is a triple of probabilistic polynomial time
 (PPT) algorithms $({\sf Gen},\,{\sf Enc},\,{\sf Dec})$ such that:
\begin{itemize}
\item {\sf Gen} is a probabilistic polynomial time key generation algorithm which takes a security parameter $1^n$ as input and
outputs a public key $pk$ and a secret-key $sk$. We write $(pk,sk) \leftarrow {\sf Gen}(1^n )$. The public key specifies the
message space ${\cal M}$ and the ciphertext space ${\cal C}$.
\item {\sf Enc} is a $($possibly$)$ probabilistic polynomial time encryption algorithm which takes as input a public key \textit{pk,} a
$m\in {\cal M}$ and random coins $r$, and outputs a ciphertext $C \in {\cal C}$. We write ${\sf Enc}(pk,m; r)$ to
indicate explicitly that the random coins \textit{r} is used and ${\sf Enc}(pk,m)$ if fresh random coins are used.
\item {\sf Dec} is a deterministic polynomial time decryption algorithm which takes as input a secret-key \textit{sk} and a ciphertext
$C \in {\cal C}$, and outputs either a message $m \in {\cal M}$ or an error symbol $\bot$. We write $m \leftarrow{\sf Dec}(C,\,sk)$.
\item $($Completeness$)$ For any pair of public and secret-keys generated by Gen and any message $m \in {\cal M}$ it holds that
${\sf Dec}(sk,\,{\sf Enc}(pk,m;r))=m$ with overwhelming probability over the randomness used by {\sf Gen} and the random coins \textit{r} used by {\sf Enc}.\\
\end{itemize}
\end{de}
\begin{de} [{\bf CCA2 security}] A public-key encryption scheme is secure against adaptive chosen-ciphertext attacks (i.e. CCA2-secure) if the advantage of any two-stage PPT
adversary ${\cal A} = ({\cal A}_1 ,\,{\cal A}_2 )$ in the following experiment is negligible in the security parameter $k$:
\begin{eqnarray*}
&&  {\bf Exp}_{{\rm PKE}, {\cal A}}^{\sf cca2}(k):\\
&&\quad \quad(pk, sk)\leftarrow {\sf Gen}(1^k)\\
&&\quad \quad(m_0, m_1, {\sf state})\leftarrow {\cal A}_{1}^{{\sf Dec}(sk,\cdot)} (pk) \quad {\rm s.t.} \quad |m_0|=|m_1|\\
&&\quad \quad b\leftarrow \{0, 1\}\\
&&\quad \quad C^{*}\leftarrow {\sf Enc}(pk, m_b)\\
&&\quad \quad b'\leftarrow {\cal A}_{2}^{{\sf Dec}(sk, \cdot)}(C^{*}, {\sf state})\\
&&\quad \quad if \, b \,= \, b^{'} \, return \, 1,\, else\, return\, 0.
\end{eqnarray*}
The attacker may query a decryption oracle with a ciphertext $C$ at any point during its execution, with the exception that ${\cal A}_2 $ is not allowed to query
${\rm Dec} (sk,\,.)$ with {\rm ``challenge"} ciphertext $ C^{*} $. The decryption oracle returns $ b^{'}  \leftarrow {\cal A}_2^{{\rm Dec} (sk,\cdot)} (C^{*} ,\,state)$.
The attacker wins the game if $b = b'$ and the probability of this event is defined as $ \Pr [{\rm Exp} \,_{{\rm PKE} ,\,{\cal A}}^{cca2} \,(k)] $.
 We define the advantage of $ {\cal A}$ in the experiment as
\begin{equation}
{\sf Adv}_{{\rm PKE} ,{\cal A}}^{\sf Ind-cca2} \,(k) = \left| \Pr [{\rm Exp}_{{\rm PKE} ,{\cal A}}^{\sf cca2} \,(k) = 1] - \frac{1}{2} \right|.
\label{eq1}
\end{equation}
\end{de}
\subsection{The McEliece PKE scheme}\label{ssec2.3}
The McEliece PKE consists of a triplet of probabilistic polynomial time algorithms $({\sf Gen}_{\rm McE}, {\sf Enc}_{\rm McE}, {\sf Dec}_{\rm McE})$.
\begin{itemize}
    \item{\bf System parameters.} $q, n, t \in \mathbb{N}$, where $t \ll n$.
    \item{\bf Key Generation.} ${\sf Gen}_{\rm McE}$ take as input security parameter $1^k$ and generate the following matrices:
\begin{enumerate}
  \item A $k \times n$ generator matrix {\sf G} of a code ${\cal G}$ over $\mathbb{F}_q$ of dimension $k$ and minimum distance $d \geq 2t+1$. (A binary irreducible Goppa
code in the original proposal).
  \item A $k \times k$ random binary non-singular matrix {\sf S}
  \item A $n \times n$ random permutation matrix {\sf P}.
\end{enumerate}
Then, ${\sf Gen}$ compute the $k \times n$ matrix ${\sf G}^{pub} = {\sf SGP}$ and outputs a public key $pk$ and a secret key $sk$, where
\[
     pk= ({\sf G}^{\rm pub},t) \quad {\rm and} \quad pk=({\sf S}, D_{\cal G}, {\sf P})
\]
where $D_{\cal G}$ is an efficient decoding algorithm for ${\cal G}$.
\item {\bf Encryption.} ${\sf Enc}_{\rm McE}(pk)$ takes plaintext ${\bf m} \in \mathbb{F}_2^k$ as input and randomly choose a vector ${\bf e}\in \mathbb{F}_2^n$ wit Hamming weight ${\sf wt}({\bf e})=t$
and computes the ciphertext {\bf c} as follows.
\[
{\bf c} = {\bf m}{\sf G}^{\rm pub}\oplus {\bf e}.
\]
\item {\bf Decryption.} To decrypt a ciphertext {\bf c},  ${\sf Dec}_{\rm McE}(sk, {\bf c})$ first calculates
\[
{\bf c}{\sf P}^{-1} = ({\bf m}{\sf S}){\sf G} \oplus {\bf e}{\sf P}^{-1}
\]
and then apply the decoding algorithm $D_{\cal G}$ to it. If the decoding succeeds, output
\[
{\bf m}=({\bf m} {\sf S}){\sf S}^{-1}.
\]
Otherwise, output $\perp$.
\end{itemize}
There are two computational assumptions underlying the security of the McEliece scheme.

{\bf Assumption 1 (Indistinguishability\footnote{This statement is not true in general. See \cite{7,8} for instance.}).} {\it The matrix {\sf G} output by ${\sf Gen}$ is computationally indistinguishable from a uniformly chosen matrix of the same size.}

{\bf Assumption 2 (Decoding hardness).} {\it Decoding a random linear code with parameters $n, k, w$ is hard.}

Note that Assumption 2 is in fact equivalent to assuming the hardness of GDP. It is immediately clear that the following corollary is true.

{\bf Corollary 1.} {\it Given that both the above assumptions hold, the McEliece cryptosystem is one-way secure under passive attacks.}
\section{The proposed cryptosystem}\label{sec3}
In this section, we introduce our conversion. Our construction consists of two parts: 1) Encryption of random coins $r$ using the original McEliece PKE scheme; 2) Randomized encoding of the plaintext, where randomization is done using $r$ (that used for consistency check) based on a heuristic encoding algorithm. Encoding includes
a permutation and combination on the message bits that performs using an algorithm called {\it permutation combination algorithm} (PCA).
\subsection{PCA encoding algorithm}
To encode message $\textbf{m}\in\{0,1\}^n$ with $n\gg k$, we firstly pick coins $\textbf{r}\in\{0,1\}^k, \textbf{r}\ne 0^k,1^k$ uniformly at random, where $k$ is the security parameter. Let ${\sf wt}(\textbf{r})=h$ be its Hamming weight. We divide $\textbf{m}$ into $l$ blocks $(b_1 \| \ldots \|b_l)$ with equal binary length $\lceil n/l\rceil$, where $l=h$ if $h\geq k-h$, else $l= k-h$. If $l\nmid n$, then we should pad $\textbf{m}$. In such cases, we can sample a random binary string (${\sf RBS}$) from $\textbf{r}$, say ${\sf RBS}={\sf Msb}_{l\lceil n / l \rceil-n}(\textbf{r})$, and pad it on the right of $\textbf{m}$\footnote{Note that since $l\in [\lceil k/2\rceil, k-1]$, the length of sampled {\sf RBS}, i.e. $l\lceil n / l \rceil-n$, is smaller than $k$, i.e. the length of $\textbf{r}$. Therefor, for all cases we do not have any problem for sampling RBS from $\textbf{r}$.}. Therefore, if $l\mid n$ then $v = n/l$, ${\sf RBS}= \varphi$ (the empty set) and $b_l ={\sf Lsb}_v\,(\textbf{m})$, else $v = \lceil n/l\rceil$, ${\sf RBS}$ is a random string with length $ l\left\lceil n/l \right\rceil-n$ which sampled from $\textbf{r}$ and $d_l={\sf Lsb}_{(n-(l-1)\left\lceil n/l \right\rceil)}\,(\textbf{m}) \|{\sf RBS}$. Now, we perform a secure permutation on the message blocks $b_i, 1\leq i\leq l$ with the following algorithm.

First, note that any positive integer $s, 1 \leq s\leq l! - 1$ uniquely can be shown as
\[
s=u_1\times (l-1)!+u_2\times (l-2)!+\dots+ u_l\times 0,\quad 0\le u_i\le l-i.
\]
Note that based on this definition we have $u_l=0$. The sequence $U_s=(u_1, \ldots, u_l)$ is called {\it factorial carry value} of $s$. We define original sequence $\textbf{m}^0$ as $\textbf{m}^0=(b_1\| \dots \|b_l)$. Recombine all elements of the
original sequence $\textbf{m}^0$ obtain $l!-1$ new sequences $\textbf{m}^1, \ldots ,\textbf{m}^{{(l! -1)}}$, which any sequence owns a corresponding factorial carry value. Using the factorial carry value of $s$, we can
efficiently obtain any sequence $\textbf{m}^{s}, 1 \leq s\leq l! -1$
with the following algorithm.
\begin{alg}[{\bf PCA encoding algorithm}] \label{alg1}
\end{alg}
\textbf{Input:} Message $\textbf{m}\in\{0,1\}^n$, coins $\textbf{r}\in \{0,1\}^k$ and integer $s, 1 \leq s\leq l! -1$.\\
\linebreak
{\bf Output:} Encoded message $\textbf{y}'=\textbf{m}^{s}=(b_1'\| \ldots \| b_l')$.\\
\linebreak
\textsc{Setup:}
\begin{enumerate}
\item $h \leftarrow {\sf wt}(\textbf{r})$. If $2h\geq k$ then $l\leftarrow h$, else $l\leftarrow k-h$.
\item If ‎$l\mid n$ then ‎set ${\sf RBS}= \varphi$; otherwise, ${\sf RBS}\leftarrow {\sf MSb}_{(l\lceil n /l\rceil-n)}\, (\textbf{r})$.
\item $\textbf{m}' \leftarrow \textbf{m}\| {\sf RBS}$ and divide $\textbf{m}'$ into $l$ blocks $(b_1\| \ldots \| b_l)$ with equal length $v=\lceil n/l\rceil$.
\end{enumerate}
\textsc{Permutation:}
\begin{enumerate}
  \item Write $s$ as $s=\sum_{i=1}^{l-1} {u_i\left( l-i \right)!}+u_l\times 0, \quad 0\le u_i\le l-i.$
  \item For $1\leq i\leq l$:\\
     \vspace{0.cm} \quad If $u_i=0$, then $b_i'\leftarrow b_i$;\\
     \vspace{0.cm} \quad Else, $b_i'\leftarrow b_{i+{u_i}}$, and for $1\leq j\leq u_i$: \\
     \vspace{0.cm} \hspace{1.1 cm} $b_{i+j}'\leftarrow b_i$;
  \item Return $\textbf{y}'=\textbf{m}^s=(b_1'\| \ldots\| b_l')$.
\end{enumerate}
Note that the {\it number} and the {\it length} of the message blocks are variable and changed by $\textbf{r}$.

It is clear that the above encoding algorithm satisfies correctness. Namely, for any $(\textbf{m}, \textbf{r})$ we have
\[
\forall \textbf{m} \in \{0, 1\}^n, \textbf{r} \in \{0, 1\}^k \ {\rm and}\ s\in \mathbb{N} : {\sf PCA}^{-1}({\sf PCA}(\textbf{m},\textbf{r},s),\textbf{r},s)=\textbf{m}.
\]

We illustrate ${\sf PCA}$ encoding algorithm with a small example. Suppose $\textbf{m} = (m_1 , \ldots ,m_{512})$ and $\textbf{r} = (r_1, \ldots, r_{25})$ with
$h = \sum_{i = 1}^{25} {r_i}=12$. Since $2h<k$, thus $l=k-h=13$. Therefore, the algorithm divides $\textbf{m}$ into 13 blocks with equal length $v=\lceil n/l\rceil=\lceil 512/13\rceil=40$. In this case, we have
to sample a string with length $l\lceil n/l\rceil-n=8$ from $\textbf{r}$ and pad it on the right of $\textbf{m}$. Therefore, we have $\textbf{m}' = (\underbrace {m_1 ,\dots,\,m_{40} }_{b_1}\| \underbrace {m_{41} ,\dots,\,m_{80} }_{b_2}\|\ldots \|\underbrace {m_{481},\dots, m_{512} \|r_1,\dots r_8 }_{b_{13}})$.\\
We choose integer $s, 1\leq s\leq 13!-1$, say $s= 4819995015$. We have
\begin{eqnarray*}
&&4819995015=10\times12!+0\times11!+8\times 10!+2\times9!+5\times 8!+4\times 7!+1\times 6!\\
&& \hspace{2.27 cm} + 3\times 5!+0\times 4!+2\times 3!+1\times2!+1\times1!+0
\end{eqnarray*}
Thus, the factorial carry value of $\textbf{m}^{4819995015}$ is $\{10,0,8,2,5,4,1,3,0,2,1,1,0\}$.
Compute sequence $D^{4819995015}$ with its factorial carry value $\{10,0,8,2,5,4,1,3,0,2,1,1,0\}$. We have
\begin{eqnarray*}
&&10--\{b_1, b_2,b_3,b_4,b_5,b_6,b_7,b_8,b_9,b_{10},b_{11},b_{12},b_{13}\}\,\, \rightarrow b_{11}\\
&&0--\{b_1, b_2,b_3,b_4,b_5,b_6,b_7,b_8,b_9,b_{10},b_{12},b_{13}\}\, \, \rightarrow b_1\\
&&8--\{b_2,b_3,b_4,b_5,b_6,b_7,b_8,b_9,b_{10},b_{12},b_{13}\}\, \, \rightarrow b_{10}\\
&&2--\{b_2,b_3,b_4,b_5,b_6,b_7,b_8,b_9,b_{12},b_{13}\}\,\,\rightarrow b_4\\
&& \quad \ \ \vdots\\
&&1--\{b_5,b_6,b_{13}\}\, \, \rightarrow b_6\\
&&1--\{b_5,b_{13}\}\, \, \rightarrow b_{13}\\
&&0--\{b_5\} \,\, \rightarrow b_5
\end{eqnarray*}

Therefore, the permutation of sequence $\textbf{m}^{4819995015}$ is $(b_{11}\| b_1\|b_{10}\|b_4\|b_8\|b_7\| b_3\|b_9\|b_2\|b_{12}\|b_6\|b_{13}\|b_5)$.
\subsection{The proposed scheme}
Now, we are ready to define our conversion. Given a McEliece PKE scheme ${\rm \Pi}=({\sf Gen}_{\rm McE}, {\sf Enc}_{\rm McE}, {\sf Dec}_{\rm McE})$, we transform it into CCA2-secure PKE scheme ${\rm \Pi'}=({\sf Gen}_{cca2}, {\sf Enc}_{cca2}, {\sf Dec}_{cca2})$.\\
\linebreak
{\bf Key Generation.} On security parameter $k$, ${\sf Gen}_{cca2}(1^k)$ run $(sk_{\rm McE},pk_{\rm McE})\leftarrow{\sf Gen}_{\rm McE}(1^k)$ to obtain $sk_{\rm McE}=({\sf S}, D_{\cal G},{\sf P})$ and $pk_{\rm McE}=({\sf G}^{pub}, t)$ as in subsection {\ref{ssec2.3}. It also choose target collision resistant (aka universal one-way) hash function ${\sf T}: \{0,1\}^k \to \{0,1\}^k$ and pseudorandom generator $G: \{0,1\}^k \to \{0,1\}^n$. $pk=\{({\sf S}, D_{\cal G},{\sf P}),G, {\sf T}\}$ is the public key and $sk=\{({\sf G}^{pub}, t),G, {\sf T}\}$ is the secret key.\\
\linebreak
{\bf Encryption.} To create ciphertext, encryption algorithm in some cases performs operations in decimal, and in other cases it performs operations in binary representation of the components. When we do operations in decimal, we show components by regular small fonts, and when we perform operations in binary, we show components by bold face fonts. To encrypt message $\textbf{m} \in \{0,\,1\}^n$ with $n\gg k$, ${\sf Enc}_{cca2}$:
\begin{enumerate}
\item Randomly choose error vector $\textbf{e}\in \{0,1\}^k$ wit Hamming weight ${\sf wt}(\textbf{e}) = t$.
\item Compute $\textbf{r}={\sf T}(\textbf{e})\in\{0,1\}^k, \textbf{r}\ne 0^k, 1^k$. Let ${\sf wt}(\textbf{r})=h$. If $2h>k$ then $l\leftarrow h$, else $l\leftarrow k-h$.
\item Compute $\tilde{\textbf{m}}=\textbf{m}\oplus G(\textbf{r})$.
\item Set $s=\sum_{i=1}^{l-1}u_i(l-i)!+u_l\times 0$, where $u_i=(r+i)\mod (l-i), 1\leq i\leq l-1$ and $u_l=0$, and run PCA encoding algorithm (Algorithm \ref{alg1}) on inputs $(\tilde{\textbf{m}},\textbf{r}, s)$ to generate encoded message $\textbf{y}'={\sf PCA}(\tilde{\textbf{m}},\textbf{r},s)$. Note that we have $U_s=(u_1, \dots, u_{l-1}, 0)$.
\item Perform a circular $q$-shift on the encoded message $\textbf{y}'$ and compute sequence $\textbf{y}={\sf CS}_{q,|\textbf{y}'|}(\textbf{y}')$, where $q= r \mod n$. Note $n$ and so $q< |\textbf{y}'|$.

\item Compute
\end{enumerate}
\[
C_1=(hy+\bar{r})r+z , \quad C_2={\sf Enc}_{\rm McE}(pk, \textbf{r};\textbf{e})=\textbf{r}{\sf G}^{pub}\oplus \textbf{e},
\]
where $r, y$ are the corresponding decimal value of $\textbf{r}$ and $\textbf{y}$. $\bar{r}$ is the decimal value of the complement of $\textbf{r}$ and $z=\sum_{i=1}^{l-1}u_i$.

As we know, in hybrid PKE schemes XOR alone cannot perfectly hide challenge bit to the CCA2 adversary. To handle CCA2 security related issues, we increase \textit{obfuscation} of the XORed message by a) perform a randomized encoding on its bits and b) perform a secure circular shift on the bits of encoded message, whose shift step depends on the value of $r$. Moreover, we disguise encoded message and conceal its bits by setting $C_1= (hy+\bar{r})r+z$ in order to \textit{decrease malleability} of the ciphertext. (See Proposition \ref{pro1}). Therefore, the CCA2 adversary to extract any useful information about challenge bit from $C_1$ must first recover the \textit{same} coins $\textbf{r}$ that was used to create the ciphertext from the McEliece PKE scheme, which is impossible if the McEliece PKE scheme be secure.\\
\linebreak
{\bf Decryption.}
To recover message $\textbf{m}$ from $C = (C_1, C_2 )$, ${\sf Dec}_{cca2}$ perform the following steps.
\begin{enumerate}
  \item Compute coins $\textbf{r}$ as $\textbf{r}= {\sf Dec}_{\rm McE}(sk_{\rm McE},C_2)$ and retrieve error vector $\textbf{e}=C_2\oplus \textbf{r}{\sf G}^{pub}$
\item Check whether
\begin{equation}\label{eq2}
\textbf{r}\stackrel{?}{=}{\sf T}(\textbf{e})
 \end{equation} holds\footnote{In deterministic code-based PKE schemes such as Niederreiter PKE scheme, we don't need to perform this checking. In these schemes since encryption algorithm is deterministic, each message has \textit{one} pre-image. Therefore if $C_2\ne C_2^*$, then $\textbf{r}={\sf Dec}(C_2, sk)\ne {\sf Dec}(C_2^*, sk)=\textbf{r}^*$. But in the McEliece encryption scheme, for $i$-th and $j$-th unit vectors $e_i$ and $e_j$ (with $i \neq j$) if ${\sf wt}(e , e_i) = 1$ and ${\sf wt}(e , e_j)=0$, then $C_2' = (\textbf{r}{\sf G}^{pub}\oplus \textbf{e})\oplus e_i \oplus e_j$ is a correct ciphertext, since the Hamming weight of $\textbf{e} \oplus e_i \oplus e_j$ is $t$. Therefore, queried ciphertext of the form $(C_1, C_2\oplus e_i\oplus e_j)$ may leaks information of the original message. Thus, we need to check well-formedness of the ciphertext and reject
such maliciously-formed one.} and reject if not ({\it consistency} check). If it holds compute ${\sf wt}(r)=h$. If $2h\geq k$ then $l\leftarrow h$, else $l\leftarrow k-h$.
\item Compute $s=\sum_{i=1}^{l-1}u_i(l-i)!$, $U_s=(u_1, \dots, u_{l-1}, 0)$ and $z=\sum_{i=1}^{l-1}u_i$, where $u_i=(r+i)\mod (l-i)$.
\item Compute
\begin{equation}\label{eq3}
y = \frac{(C_1-z)/r-\bar{r}}{h},
 \end{equation}
and reject if $y$ is not a $(l\lceil n/l\rceil)$-bit integer (consistency check). Note that if $l\mid n$ then $|y|=n$.
\item Compute $\textbf{y}'={\sf CS}_{q,|\textbf{y}|}^{-1}(\textbf{y})$, where ${\bf y}$ is the binary representation of $y$. $q=r\mod n$ and $|\textbf{y}'|=|\textbf{y}|$.
\item Compute $\tilde{\textbf{m}}={\sf PCA}^{-1}(\textbf{y}',\textbf{r},s)$.
 \item Compute $\textbf{m}'= \tilde{\textbf{m}}\oplus G(\textbf{r})$ and ${\sf RBS}={\sf Msb}_{l\lceil n / l \rceil-n}(\textbf{r})$. Check wether
 \begin{equation}\label{eq4}
{\sf Lsb}_{l\lceil n / l \rceil-n}(\textbf{m}')\stackrel{?}{=}{\sf RBS}
 \end{equation} holds and reject if not (consistency check). If it holds output
  \begin{equation}\label{eq5}
 \textbf{m}={\sf Msb}_{l\lceil n / l \rceil-n}(\textbf{m}'),
  \end{equation}
 else output $\bot$.
\end{enumerate}

\begin{rem} \label{rem1}
Note that if either $\textbf{y}$ or $\textbf{r}$ be illegal values, then the decryption algorithm outputs a \textit{random} string.
\end{rem}
\begin{pro}\label{pro1} The ciphertext $C_1$ is non-malleable.
\end{pro}
\begin{proof} We say that $(\hat{C_1}, C_2)$ is a valid forgery on $(C_1, C_2)$ if $\hat{\textbf{y}}$ differ from $\textbf{y}$ in some bits, where $\hat{\textbf{y}}$ and $\textbf{y}$ are the corresponding binary representation of $\hat{y}= ((\hat{C_1}-z)/r-\bar{r})/h$ and $y = ((C_1-z)/r-\bar{r})/h$ respectively. If the adversary can produce such $\hat{C_1}$, then he can guess challenge bit from
$\hat{\textbf{m}}={\sf Msb}_{l\lceil n / l \rceil-n}(\hat{\tilde{\textbf{m}}}\oplus G(\textbf{r}))$. Without loss of generality, we assume $\hat{\textbf{y}}$ and $\textbf{y}$ differ in $i$-th bit. Namely $\hat{\textbf{y}}=\textbf{y}\oplus e_i$, where $e_i$ is the $i$-th unit vector. Thus, we should have $\hat{y}=y\pm 2^i$. That is, $\hat{C_1}=(h\hat{y}+\bar{r})r+z=C_1\pm 2^{i}hr$. Since secret coins $r\in\{0,1\}^k$ is not known to the adversary, thus the probability that the adversary produces a valid forgery on $C_1$ is negligible and it is $2^{-k}$. This is because in the encryption algorithm we set $C_1=(hy+\bar{r})r+z$.
\end{proof}
\begin{thm}\label{thm1}
Suppose ${\rm \Pi}=({\sf Gen}_{\rm McE},{\sf Enc}_{\rm McE}, {\sf Dec}_{\rm McE})$ be the McEliece PKE scheme. Then, the proposed scheme ${\rm \Pi}'=({\sf Gen}_{cca2},{\sf Enc}_{cca2}, {\sf Dec}_{cca2})$ is CCA2-secure in the standard model.
\end{thm}
\begin{proof} The encryption algorithm uses coins $r$ to encrypt challenge message. In the encryption algorithm, we do not use any cryptographic primitives to be able to reduce CCA security of the proposed PKE scheme to the hardness or security of them. Note that we only use the McEliece PKE scheme to encrypt the coins $r$ and encryption of the challenge message independent of its output.\\
In the proof of security, we exploit the fact that for a given ciphertext, we can recover the message if we know the \textit{same} encoded message $y$ and randomness $r$ that was used to create the ciphertext. We stress that if either $y$ or $r$ is not legal values, then the output of the decryption algorithm is \textit{random}. Thus, the challenge bit is information-theoretically hidden to the CCA2 adversary, and so, his advantage in guessing challenge bit is 0.

Let $C^*=(C_1^*,C_2^*)$ be the challenge ciphertext, where $C_1^* = (h^*y^*+\bar{r^*})r^*+z^*$ and $\textbf{y}^*={\sf PCA}(\textbf{m}_b\oplus G(\textbf{r}^*),\textbf{r}^*)$. Denote the secret randomness used to encrypt $\textbf{m}_b$ by $\textbf{r}^*$. Assume towards contradiction that there is an efficient adversary ${\cal A}$ breaking CCA2 security of the proposed PKE scheme with non-negligible probability. That is, the adversary ${\cal A}$ can guess challenge bit with non-negligible probability at least from one of the below cases. Since a decryption query on the challenge ciphertext is forbidden by the CCA2-experiment, thus if $C_1=C_1^*$, then $C_2\ne C_2^*$ and vice versa. Therefore, there are three possible cases:\\
\linebreak
\textbf{Case1.} $C=(C_1,C_2)\ne (C_1^*,C_2^*)$. In this case, the decryption oracle takes as input $(C_1,C_2)$ and compute $\textbf{r}={\sf Dec}_{\rm McE}(C_2)\in \{0,1\}^k$. If $\textbf{r}=\textbf{r}^*$ while $C_2\ne C_2^*$, then the decryption oracle will reject in (\ref{eq2}). It also computes $l$, $u_i=(r+i)\mod (l-i),1\leq i\leq l-1$, $s=\sum_{i=1}^{l-1}u_i(l-i)!$ and $z=\sum_{i=1}^{l-1}u_i$. In the worst case, we assume $y = ((C_1-z)/r-\bar{r})/h$ is a $(n+l\lceil n/l\rceil-n)$-bit integer. That is, there is an integer $y$ such that $C_1=(hy+\bar{r})r+z$. The decryption oracle computes $\textbf{y}'={\sf CS}_{r \mod n,|\textbf{y}|}^{-1}(\textbf{y})$ and decodes $\textbf{y}'$ based on recovered coins $\textbf{r}$ and computed value $s$. We have $\tilde{\textbf{m}}={\sf PCA}^{-1}(\textbf{y}',\textbf{r},s)\ne{\sf PCA}^{-1}(\textbf{y}'^*,\textbf{r}^*,s^*)=\tilde{\textbf{m}}^*$, even we assume $\textbf{y}'=\textbf{y}'^*$. If we also assume condition (\ref{eq4}) holds (i.e., ${\sf Lsb}_{l\lceil n / l \rceil-n}(\tilde{\textbf{m}}\oplus G(\textbf{r}))={\sf Msb}_{l\lceil n / l \rceil-n}(\textbf{r})$), then the decryption oracle outputs random string $\textbf{m}={\sf Msb}_{l\lceil n / l \rceil-n}(\tilde{\textbf{m}}\oplus G(\textbf{r}))$ in (\ref{eq5}). Since $\textbf{m}$ is a random string, thus challenge bit is information-theoretically hidden to the CCA2 adversary, and so, his advantage to guess challenge bit is 0.\\
\linebreak
\textbf{Case2.} $C=(C_1^*,C_2\ne C_2^*)$. In this case, the decryption oracle takes as input $(C_1^*,C_2)$ and compute $\textbf{r}={\sf Dec}_{\rm McE}(C_2)$. If $\textbf{r}=\textbf{r}^*$ while $C_2\ne C_2^*$, then the decryption oracle will reject in (\ref{eq2}). In the worst case, we assume $y=((C_1^*-z)/r-\bar{r})/h$ is a $(n+l\lceil n/l\rceil-n)$-bit integer and $\textbf{y}'={\sf CS}_{r \mod n,|\textbf{y}|}^{-1}(\textbf{y})=\textbf{y}'^*$. Since $\textbf{r}$ is illegal, i.e. $\textbf{r}\ne \textbf{r}^*$ (and so $s\ne s^*$), thus $\tilde{\textbf{m}}={\sf PCA}^{-1}(\textbf{y}',\textbf{r},s)\ne{\sf PCA}^{-1}(\textbf{y}'^*,\textbf{r}^*,s^*)=\tilde{\textbf{m}}^*$. If we also assume condition (\ref{eq4}) is hold, then the decryption algorithm outputs random string $\textbf{m}={\sf Msb}_{l\lceil n / l \rceil-n}(\tilde{\textbf{m}}\oplus G(\textbf{r}))$ in (\ref{eq5}). Therefore, challenge bit is information-theoretically hidden to the CCA2 adversary, and so, his advantage to guess challenge bit is 0.\\
\linebreak
\textbf{Case3.} $C=(C_1\ne C_1^*,C_2^*)$. In this case, the decryption oracle takes as input $(C_1,C_2^*)$ and compute $\textbf{r}={\sf Dec}_{\rm McE}(C_2^*)=\textbf{r}^*$. It also computes $y=((C_1-z^*)/r^*-\bar{r^*})/h^*\ne y^*$ \footnote{In this case we have $y\ne y^*$. If $y= y^*$, then we have $C_1=(h^*y+\bar{r^*})r^*+z^*= (h^*y^*+\bar{r^*})r^*+z^*=C_1^*$, which is a contradiction since a decryption query on the challenge ciphertext is forbidden by the CCA2-experiment.}. In the worst case, we assume $y$ is an integer. We consider tree possible cases for $y$:
\begin{enumerate}
\item [a)] $y$ is a multiple of $y^*$. That is, for any $k \in \mathbb{N}, k\ne1$ we have $y=ky^*$. In this case $|y|=|k||y^*|\ne|y^*|=n+l\lceil n/l\rceil-n$ and the decryption oracle reject in (\ref{eq3}).
\item [b)] $|\textbf{y}|=|\textbf{y}^*|$, and, $\textbf{y}$ and $\textbf{y}^*$ are differ from each other only in some bits. In this case, $\textbf{y}'={\sf CS}_{q^* ,|\textbf{y}|}^{-1}(\textbf{y})$ and $\textbf{y}'^*={\sf CS}_{q^* ,|\textbf{y}|}^{-1}(\textbf{y})$ are also differ from each other only in some bits. Therefore, the CCA2 adversary can guess challenge bit from $\textbf{m}'={\sf PCA}^{-1}(\textbf{y}',\textbf{r}^*,s^*)\oplus G(\textbf{r}^*)$. Without loss of generality, we can assume $\textbf{y}=\textbf{y}^*\oplus e_i$, where $e_i$ is the $i$-th unit vector. Thus we have $y=y^*\pm 2^i$, where $y$ and $y^*$ are the corresponding decimal value of $\textbf{y}$ and $\textbf{y}^*$. Therefore, we should have $C_1=(y^*h^*+\bar{r^*})r^*+z^*\pm 2^{i}h^*r^*=C_1^*\pm 2^{i}h^*r^*$. Since secret coins $r^*\in \{0,1\}^k$ is not known to the CCA2 adversary, thus the probability that CCA2 adversary produce a forgery on $C_1$ is negligible and it is $2^{-k}$. Therefore, the CCA2 adversary's advantage in this case is negligible (see also Proposition\ref{pro1}).
\item [c)] $|\textbf{y}|=|\textbf{y}^*|$, and, $\textbf{y}\ne \textbf{y}^*$ (and so $\textbf{y}'= {\sf CS}_{q^* ,|\textbf{y}|}^{-1}(\textbf{y})$) is a random string. In this case $\textbf{m}'={\sf PCA}^{-1}(\textbf{y}',\textbf{r}^*,s^*)\oplus G(\textbf{r}^*)$ also is a random string and the decryption oracle outputs random string ${\sf Msb}_{l\lceil n/l\rceil-n}(\textbf{m}')$. Thus, challenge bit is information-theoretically hidden to the CCA2 adversary, and so, his advantage to guess challenge bit is 0.
\end{enumerate}

From \textbf{Case1}, \textbf{Case2} and \textbf{Case3}, the CCA2 adversary advantage to guess challenge bit is negligible. This contradicts the assumption that the CCA2 adversary can break CCA2 security of the proposed PKE scheme with non-negligible probability.
\end{proof}
\subsection{Performance analysis}
The performance-related issues can be discussed with respect to the computational complexity of key generation, key sizes, encryption and decryption speed.

The resulting encryption scheme is very efficient. The public/secret keys are roughly as
in the original scheme. The time for computing ${\sf T}(\cdot), G(\cdot)$ and the time for encoding and decoding is negligible compared to the time for computing ${\sf Enc}_{\rm McE}$ and ${\sf Dec}_{\rm McE}$. Encryption roughly needs one application of ${\sf Enc}_{\rm McE}$, and decryption roughly needs one application of
${\sf Dec}_{\rm McE}$.

As we previously stated, the Niederreiter-based proposed scheme does not need to perform well-formedness checking. Therefore, compared to Freeman \textit{et al.}\cite{10} and Mathew \textit{et al.}\cite{19} schemes, our scheme is more efficient. The comparison of the proposed schemes with existing schemes are presented in Table 1.
\begin{center}
{\small {\bf Table 1.} Comparison with other proposed CCA2-secure code-based cryptosystems}
\begin{tabular}{|c|c|c|c|c|c|}
  \hline
  {\bf Scheme} & {\bf Public-key} & {\bf Secret key} & {\bf Ciphertext} & {\bf Encryption} & {\bf Decryption} \\
 &  & & Size  & {\bf Complexity} & {\bf complexity} \\
\hline
  Dowsley  and& $2k \times pk_{\rm McE}$ & $2k \times sk_{\rm McE}$ & $k \times {\sf Ciph}_{\rm McE}$ & $k \times {\sf Enc}_{\rm McE}+$ & $1\, {\sf Ver}_{\cal {OT-SS}}$+\\
 D\"{o}ttling &  & &  &  & \\
\textit{et al.}\cite{5,6} &  &  &  &$1\, {\cal OT-SS}$ & $1 \times {\sf Dec}_{\rm McE}$+\\
&  & &  &  &  $t \times {\sf Enc}_{\rm McE}$\\
&  & &  &  &  \\
Freeman&$2k \times pk_{\rm Nie}$ & $2k \times sk_{\rm Nie}$ & $k \times {\sf Ciph}_{\rm Nie}$  & $k \times {\sf Enc}_{\rm Nie}+$ &$1\, {\sf Ver}_{\cal {OT-SS}}+$\\
\textit{et al.}\cite{10} & & & &$1\,{\cal OT-SS}$ &$1 \times {\sf Dec}_{\rm Nie}$+\\
&  & &  &  &  $t \times {\sf Enc}_{\rm Nie}$\\
&  & &  &  &  \\
Mathew & $1 \, pk_{\rm Nie}+$ & $2 \times sk_{\rm Nie}$ & $2 \times {\sf Ciph}_{\rm Nie}$ & $2 \times {\sf Enc}_{\rm Nie}$+ &$1\, {\sf Ver}_{\cal {OT-SS}}+$\\
 \textit{et al.}\cite{19}& $1\,(n\times n)$ &&&1 {\sf MM}+&$1 \times {\sf Dec}_{\rm Nie}$+\\
& {\rm Matrix} & &  &$1\, {\cal OT-SS}$ &$2 \times {\sf Enc}_{\rm Nie}$+ \\
&  & &  &  &1 {\sf MM}\\
&  & &  &  & \\
Proposed &$\thickapprox 1\, pk_{\rm McE}$ & $\thickapprox 1\, sk_{\rm McE}$ & $\thickapprox2 \, {\sf Ciph}_{\rm McE}$ & $\thickapprox 1 \, {\sf Enc}_{\rm McE}$ & $\thickapprox 1 \,{\sf Dec}_{\rm McE} $\\
Scheme & & & $+n$ & & \\
\hline
\end{tabular}
\end{center}
McE: McEliece cryptosystem, Nie: Niederreiter cryptosystem, {\sf Ciph}: Ciphertext, {\sf Ver}: Verification, ${\cal OT-SS}$: Strongly unforgeable one-time signature scheme, P: Product,
D: Division, {\sf MM}: Matrix Multiplication, {\sf PCA}: Permutation Combination Algorithm (algorithm {\ref{alg1}}), ${\sf PCA}^{-1}$: Reverse Permutation Combination Algorithm and $t\leq k$.\\
\section{General construction from TDFs} \label{sec4}
Devising public-key encryption schemes which are secure against chosen ciphertext attack from low-level primitives has been the subject of investigation by many researchers. Currently,
the minimal security assumption on trapdoor functions need to obtain CCA2-secure PKE schemes, in terms of ``black-box" implications, is that of {\it adaptivity} was proposed by Kiltz,
Mohassel and O'Neill in Eurocrypt 2010 \cite{11}. They proposed a black-box \textit{one}-bit CCA2-secure encryption scheme and then apply a transform of Myers and shelat \cite{16} from one-bit to multi-bit CCA-secure encryption scheme. The Myers-shelat conversion is not efficient; it uses encryption reputation paradigm along with a strongly unforgeable one-time
signature scheme to handle CCA2 security related issues. Therefore, the resulting encryption scheme needs \textit{separate} encryption and it is not sufficiently efficient to be used
in practice.

Here, we give \textit{direct} black-box construction of a CCA2-Secure PKE scheme from TDFs. Our construction is similar to the construction of Section \ref{sec3}. We only need to replace the underlying code-based PKE scheme with a OW-TDF. Let ${\sf TDF = (Tdg, F, F^{-1})}$ be an {\it injective} TDF. We construct multi-bit PKE scheme ${\sf PKE[TDF] =(Gen, Enc, Dec)}$ as follows:\\
\linebreak
{\bf Key Generation.} On security parameter $k$, the generator ${\sf Gen}$ runs ${\sf Tdg}$ to obtain $(ek,td)\leftarrow {\sf Tdg}(1^k)$ and return $(ek,td)$. It also chooses PRG $G: \{0,1\}^k\to \{0,1\}^n$. $pk=(ek, G)$ is the public key and $sk=(td,G)$ is the secret key.\\
\linebreak
{\bf Encryption.} On inputs $(\textbf{m}, ek)$, where $\textbf{m} \in \{0,\,1\}^n$, ${\sf Enc}$ perform as follows:
\begin{enumerate}
\item Choose coins $\textbf{r}\in \{0,1\}^k, \textbf{r}\ne 0, 1^k$ uniformly at random and let ${\sf wt}(\textbf{r})=h$. If $2h>k$ then $l\leftarrow h$, else $l\leftarrow k-h$.
\item Compute $\tilde{\textbf{m}}=\textbf{m}\oplus G(\textbf{r})$.
\item Set $s=\sum_{i=1}^{l-1}u_i(l-i)!+u_l\times 0$, where $u_i=(r+i)\mod (l-i), 1\leq i\leq l-1$ and $u_l=0$, and run ${\sf PCA}$ encodding algorithm (Algorithm \ref{alg1}) on inputs $(\tilde{\textbf{m}},\textbf{r}, s)$ to generate encoded message $\textbf{y}'={\sf PCA}(\tilde{\textbf{m}},\textbf{r},s)$.
\item Perform a circular $q$-shift on the encoded message $\textbf{y}'$ and compute $\textbf{y}={\sf CS}_{q,|\textbf{y}'|}(\textbf{y}')$, where $q= r \mod n$.
\item Compute
\end{enumerate}
\[
C_1  = (hy+\bar{r})r+z , \quad C_2={\sf F}(ek, \textbf{r}),
\]
where $\bar{r}$ is the decimal value of the complement of $\textbf{r}$ and $z=\sum_{i=1}^{l-1}u_i$.\\
\linebreak
{\bf Decryption.} On inputs $(C, td)$, ${\sf Dec}$ perform as follows:
\begin{enumerate}
  \item Compute coins $\textbf{r}$ as $\textbf{r}= {\sf F}^{-1}(C_2, td)$. Compute $\bar{r}$ and $h={\sf wt}(r)$. If $h\geq k-h$ then $l\leftarrow h$, else $l\leftarrow k-h$.
\item Compute $U_s=(u_1, \dots, u_{l-1}, 0)$ and $z=\sum_{i=1}^{l-1}u_i$, where $u_i=(r+i)\mod (l-i)$ and $s=\sum_{i=1}^{l-1}u_i(l-i)!$.
\item Compute
\[
y = \frac{(C_1-z)/r-\bar{r}}{h},
\]
and reject the ciphertext if $y$ is not a ($l\lceil n/l\rceil$)-bit integer.
  \item Compute $\textbf{y}'={\sf CS}_{q,|\textbf{y}|}^{-1}(\textbf{y})$. $q=r\mod n$ and $|\textbf{y}'|=|\textbf{y}|$.
\item Compute $\tilde{\textbf{m}}={\sf PCA}^{-1}(\textbf{y}',\textbf{r},s)$.
 \item Compute $\textbf{m}'= \tilde{\textbf{m}}\oplus G(\textbf{r})$ and ${\sf RBS}={\sf Msb}_{l\lceil n / l \rceil-n}(\textbf{r})$. Check wether
 \begin{equation}\label{eq4}
{\sf Lsb}_{l\lceil n / l \rceil-n}(\textbf{m}')\stackrel{?}{=}{\sf RBS}
 \end{equation} holds and reject if not. If it holds output
  \begin{equation}\label{eq5}
 \textbf{m}={\sf Msb}_{l\lceil n / l \rceil-n}(\textbf{m}'),
  \end{equation}
 else output $\bot$.
\end{enumerate}

\begin{thm}\label{thm2}
Let ${\sf TDF}$ be a one-way trapdoor function, then the ${\sf PKE[TDF]}$ defined above is CCA2-secure.\\
The proof of Theorem\ref{thm2} is similar to the proof of Theorem\ref{thm1} which is omitted.
\end{thm}

\end{document}